\documentclass[12pt,onecolumn,draftcls,journal,letterpaper]{IEEEtran}

\usepackage{amsthm} 

\usepackage{amssymb,amsmath,color} 
\usepackage{graphicx} 
\usepackage{epsfig}
\usepackage{algorithm}
\usepackage[noend]{algpseudocode}

\usepackage{cite}

\renewcommand{\natural}{{\mathbb{N}}} 

\newcommand{\real}{{\mathbb{R}}}
 
\newcommand{\union}{\cup}
\newcommand{\intersection}{\ensuremath{\operatorname{\cap}}}

\newcommand{\map}[3]{#1: #2 \rightarrow #3}

\newcommand{\until}[1]{\{1,\ldots,#1\}}

\newcommand{\CC}{\mathcal{C}} 
\newcommand{\EE}{\mathcal{E}} 
\newcommand{\GG}{\mathcal{G}}

\newcommand{\prox}{\mathbf{prox}}
\newcommand{\argmin}{\mathop{\rm argmin}}

\newcommand{\nbrs}{\mathcal{N}}

\algdef{SE}[DOWHILE]{Do}{doWhile}{\hskip\algorithmicindent\algorithmicdo}[1]{\hskip\algorithmicindent\algorithmicwhile\ #1}%

\makeatletter
\newcommand{\StatexIndent}[1][3]{%
  \setlength\@tempdima{\algorithmicindent}%
  \Statex\hskip\dimexpr#1\@tempdima\relax}
\makeatother

\algnewcommand{\algorithmicgoto}{\textbf{go to }}%
\algnewcommand{\Goto}[1]{\algorithmicgoto Line~\ref{#1}}%
\algnewcommand{\Label}{\State\unskip}

\renewcommand{\algorithmicwhile}{\hskip\algorithmicindent \textbf{While:}}

\newcommand{\sx}[1]{\bar{x}_{#1}} 

\usepackage{soul} 

\newtheorem{theorem}{Theorem}[section]

 \newtheorem{lemma}[theorem]{Lemma}
\newtheorem{remark}[theorem]{Remark}

\newtheorem{assumption}[theorem]{Assumption}
 
\def \algnamefull/{Coordinate Paritioned Descent}
\def \algname/{CPD}

\newcommand\oprocendsymbol{\hbox{$\square$}}
\newcommand\oprocend{\relax\ifmmode\else\unskip\hfill\fi\oprocendsymbol}

\usepackage{tikz}
\usetikzlibrary{shapes,arrows}

\graphicspath{{figs/}}

\def \asynch_alg/{Asynchronous Distributed Proximal Gradient}

\begin{document}


\title{A randomized primal distributed algorithm\\ for partitioned and
  big-data non-convex optimization}


\author{Ivano Notarnicola and Giuseppe Notarstefano
  \thanks{Ivano Notarnicola and Giuseppe Notarstefano are with the Department of Engineering,
    Universit\`a del Salento, Via Monteroni, 73100
    Lecce, Italy, \texttt{name.lastname@unisalento.it.} This result is part of a
    project that has received funding from the European Research Council (ERC)
    under the European Union's Horizon 2020 research and innovation programme
    (grant agreement No 638992 - OPT4SMART).  } }

\maketitle

\begin{abstract}
  In this paper we consider a distributed optimization scenario in which the
  aggregate objective function to minimize is partitioned, big-data and possibly
  non-convex. Specifically, we focus on a set-up in which the dimension of the
  decision variable depends on the network size as well as the number of local
  functions, but each local function handled by a node depends only on a (small)
  portion of the entire optimization variable. This problem set-up has been
  shown to appear in many interesting network application scenarios.  As main
  paper contribution, we develop a simple, primal distributed algorithm to solve
  the optimization problem, based on a randomized descent approach, which works
  under asynchronous gossip communication.
  We prove that the proposed asynchronous algorithm is a proper, ad-hoc version
  of a coordinate descent method and thus converges to a stationary point.
  To show the effectiveness of the proposed algorithm, we also present
  numerical simulations on a non-convex quadratic program, which confirm the
  theoretical results.


\end{abstract}

\section{Introduction}
\label{sec:intro}
In several network scenarios optimization problems arise in which an aggregate cost
function, sum of local cost functions, needs to be minimized in a distributed
way. 
A typical approach in distributed optimization is to develop
algorithms in which the processors in the network reach consensus on a minimizer
of the problem. 
However, when the dimension of the decision variable depends on the number of
agents in the network the consensus approach gives rise to algorithms which
scale badly with the network size. 
Enforcing consensus on the entire vector of decision variables is not necessary
in many important applications, since the nodes are interested in computing only
part of the decision vector, namely only some local variables of interest.
In this paper we consider a partitioned problem set-up in which the aggregate
function is the sum of local functions, each one depending only on a portion of
the decision vector. 
For this set-up our goal is to design a distributed algorithm in which the nodes
compute only a local portion of interest of the entire solution vector, so that
the whole minimizer can be obtained by stacking together the local portions.

%
%
This partitioned set-up has been introduced in~\cite{erseghe2012distributed}
where a distributed ADMM-based algorithm is proposed. 
In~\cite{carli2013distributed} some concrete motivating scenarios are
described for the same set-up and a dual decomposition algorithm is proposed. In
both the above references the algorithms are designed for a synchronous network
with a fixed communication graph.
In~\cite{necoara2016parallel}, an analogous problem formulation is considered 
within a parallel context. The authors propose a coordinate
descent method and derive its convergence rate.
In \cite{todescato2015robust} the authors propose a distributed algorithm for a
partitioned quadratic program under lossy communication.  
%
A distributed ADMM-based algorithm with applications in MPC is proposed
in~\cite{mota2015distributed} to deal with an unconstrained optimization problem
with local domains which is related to the set-up in this paper.

%
%
Usually, distributed approaches need a common clock (e.g., because a diminishing
(time-varying) step-size is used).
We want to avoid this limitation designing an asynchronous, event-triggered
protocol based on local and independent timers, \cite{notarnicola2015randomized}.
A Newton-Raphson consensus strategy is proposed in~\cite{zanella2012asynchronous}
to solve unconstrained, convex optimization problems under asynchronous, 
symmetric gossip communications.
In~\cite{dimarogonas2012distributed} a self-triggered communication protocol is 
considered. Based on an error condition a distributed, continuous-time algorithm
is developed. 
In~\cite{wei20131} an asynchronous ADMM-based distributed method is proposed for a 
separable, constrained optimization problem with a convergence rate $O(1/t)$.
A distributed, asynchronous algorithm for constrained optimization based
on random projections is proposed in~\cite{lee2013asynchronous}.

%
The asynchronous, distributed algorithm we design in this paper is based on a
(randomized) coordinate descent method.
In~\cite{nesterov2012efficiency} the coordinate method for huge scale optimization
has been introduced. 
%
This powerful approach has been extended to deal with (convex) composite
objective functions and parallel scenarios,
see~\cite{richtarik2012parallel,richtarik2014iteration,necoara2016parallel}.
In~\cite{necoara2013random} a coordinate approach to solve linearly constrained
problems has been proposed.
Using a coordinate ADMM-based approach, in~\cite{bianchi2014stochastic}
a distributed, asynchronous algorithm is developed.

%
%
Regarding non-convex optimization problems, in~\cite{patrascu2015efficient}, the
authors extend the coordinate approach to large-scale non-convex optimization
proving the rate of convergence of their algorithms.
A parallel algorithm based on local strongly convex approximations is exploited 
in~\cite{facchinei2015parallel} to cope with non-convex optimization problems.
%
%
The latter approach has been extended to a distributed context in~\cite{dilorenzo2016next}.
In~\cite{binetti2014distributed}, the authors proposed an auction-based
distributed algorithm for non-convex optimization.

%
%
As main paper contribution we propose an asynchronous, distributed algorithm to
solve partitioned, big-data non-convex optimization problems. The proposed
primal algorithm is based on local updates involving the minimization of a
strongly convex, quadratic approximation of the objective function. Each node
constructs this approximation by exchanging information only with neighboring
nodes.  The updates at each node are regulated by a local timer that triggers
independently from the ones of the other nodes.
We prove the convergence in probability of the distributed algorithm by showing
that it is equivalent to a generalized coordinate descent method for the
minimization of non-convex composite functions.
The generalized coordinate descent algorithm extends the one proposed
in~\cite{patrascu2015efficient} and thus represents a side interesting result.






The paper is organized as follows. In Section~\ref{sec:setup} we present the
problem set-up. In Section~\ref{sec:algorithm} we propose our algorithm and
prove its convergence in Section~\ref{sec:analysis}. Finally, in
Section~\ref{sec:simulations} we show some
simulations.

\paragraph*{Notation}
Consider a vector $x\in\real^n$ partitioned in $N$ block-components as follows
\begin{align}
  x= [ x_1^\top, \ldots, x_N^\top ]^\top,
  \label{eq:x_partitions}
\end{align}
where, for all $i\in \until{N}$, we have $x_i\in \real^{n_i}$ and
$\sum_{i=1}^N n_i=n$.
Moreover, consider a block decomposition of the $N\times N$ identity 
matrix $I=[U_1, \ldots, U_N]$, where for all $i\in\until{N}$ each 
$U_i\in\real^{n\times n_i}$. Then we can write $x_i = U_i^\top x$ 
and $x = \sum_{i=1}^N U_ix_i$.
For a function $\map{\varphi}{\real^n}{\real}$, we denote $\nabla_{x_i}
\varphi(\bar{x}) = U_i^\top \nabla \varphi(\bar{x})$ the ``partial'' gradient of
$\varphi$ with respect to $x_i\in\real^{n_i}$.


\section{Optimization problem set-up}
\label{sec:setup}
We consider a network of $N$ nodes which can interact according to a
fixed, undirected communication graph $\GG = (\until{N}, \EE)$,
where
$\EE\subseteq \until{N} \times \until{N}$ is the set of edges. That is, the edge
$(i,j)$ models the fact that node $i$ and $j$ can exchange information.  We
denote by $\nbrs_i$ the set of \emph{neighbors} of node $i$ in the fixed graph
$\GG$, i.e., $\nbrs_i := \left\{j \in \until{N} \mid (i,j) \in \EE \right\}$,
and by $|\nbrs_i|$ its cardinality. Here we assume that the graph contains also
self-edges, so that $\nbrs_i$ contains also $i$.

We want to stress that the fixed graph only models, for each node, the set of
possible neighbors the node can communicate with. On top of this graph, we will
consider an asynchronous communication protocol described later.

We start by a common set-up in distributed optimization, i.e., the minimization
of a separable cost function composed by two contributions, i.e.,
$\min_{x\in\real^n} \: \: \sum_{i=1}^N f_i(x) + g_i(x)$,
where $f_i : \real^n \to \real$ and $g_i : \real^n \to \real \union\{+\infty\}$,
with $N,n \in\natural$. Usually this composite structure of the objective
functions, is used to split the effective cost into a smooth part (modeling some
local objective) and a (possibly) non-smooth one being a regularization term or
a constraint.\footnote{A constraint
  $x\in \intersection_{i\in\until{N}} X_i\subset\real^n$ is modeled by setting
  $g_i(x) = I_{X_i}(x)$, with $I_{X_i}(x)=0$ $\forall x\in X_i$ and
  $I_{X_i}(x)=+\infty$ otherwise.}

In this paper we consider problems in which the composite function
has 
a \emph{partitioned structure}, that we next describe. Let the decision
variable $x\in\real^n$ be partitioned as stated in~\eqref{eq:x_partitions}, then
the sub-vector $x_i\in \real^{n_i}$ with $n_i \ll n$, represents the relevant
information at node $i$.
Each local objective $f_i$ has a sparsity consistent with the interaction graph,
namely, for $i\in \until{N}$, the function $f_i$ depends only on the component
of node $i$ and of its neighbors. To highlight this property we let
$f_i:\real^{\sum_{j \in \nbrs_i} n_j} \rightarrow \real$ and
write $f_i ( x_{\nbrs_i} )$.  Also, each function $g_i$ depends only on the
component $x_i$, i.e., $g_i:\real^{n_i} \rightarrow \real \union \{+\infty\}$.

In light of the described structure, the problem we aim at solving in a
distributed way can be written as
\begin{align}
  \min_{x \in \real^n} \: &\: \sum_{i=1}^N f_i (x_{\nbrs_i} ) + g_i(x_i),
  \label{eq:partitioned_problem}
\end{align}  
where node $i$ knows only the functions $f_i$ and $g_i$. We call this problem
\emph{partitioned} (due to the structure of the functions $f_i$ and $g_i$) and
\emph{big-data} (since the dimension of the decision variable depends on the
number of nodes).

Note that, in this partitioned scenario, network structure and objective function 
are inherently related. That is, nodes that share a variable are neighbors in
the communication graph.
As pointed out in the introduction this set-up appears in
several interesting applications~\cite{carli2013distributed}.
In the following assumptions we state the main properties of problem
\eqref{eq:partitioned_problem}.
\begin{assumption}
  For all $i\in\until{N}$, $f_i$ 
  is a smooth function of $x_{\nbrs_i}$. In particular,
  $f_i$ has block-coordinate Lipschitz continuous gradient, i.e., 
  for all $j\in \nbrs_{i}$ there exists constants $L_{ij} > 0$ such that
  for all $x_{\nbrs_i} \in \real^{\sum_{\ell \in \nbrs_i } n_\ell}$ 
  and $s_j \in \real^{n_j}$ it holds
  \begin{align*}
    \| \nabla_{x_j} f_i(x_{\nbrs_i}+U_{ij} s_j) -  \nabla_{x_j} f_i(x_{\nbrs_i}) \| \le L_{ij} \| s_j\|.
  \end{align*}%
  where $U_{ij}$ is a suitable matrix such that $U_{ij} s_j$ is a vector in
  $\real^{\sum_{\ell \in \nbrs_{i}} n_\ell }$ with $j$-th block-component
  equal to $s_j$ and all the other ones equal to zero.
  \oprocend
  \label{ass:Lipschitz_f_i}
\end{assumption}

In light of Assumption~\ref{ass:Lipschitz_f_i}, it is easy to show
that the following lemma holds.
\begin{lemma}
  Let Assumption~\ref{ass:Lipschitz_f_i} hold, then the aggregate function
  $f(x) := \sum_{i=1}^N f_i ( x_{\nbrs_i} )$ has block-coordinate Lipschitz
  continuous gradient. In particular, for all $i\in\until{N}$, the partial
  gradient $\nabla_{x_i} f$ has Lipschitz constant given by
  $L_i := \sum_{j\in \nbrs_{i} } L_{ij}$.
  %
\end{lemma}
\begin{proof}
  The proof follows straight by simply writing the norm of the aggregate cost
  $f$ and then bounding each term of its gradient by using its block Lipschitz
  constant.
\end{proof}

\begin{remark}
  Note that one can assume directly that $\nabla_{x_i} f$ is Lipschitz
  continuous, but while the condition we impose can be checked in a distributed
  way, the weaker one needs a global knowledge of the cost $f$.\oprocend
\end{remark}

\begin{assumption}
  For all $i\in\until{N}$, the function $g_i$
  is a proper, closed, proper, convex
  function. 
  \oprocend
  \label{ass:convex_g_i}
\end{assumption}

We stress that we have not assumed any convexity condition on $f_i$, thus
optimization problem~\eqref{eq:partitioned_problem} is non-convex in general.
Finally, we state the following assumption which
is quite standard for non-convex scenarios.
\begin{assumption}
  The cost $V(x):=\sum_{i=1}^N f_i (x_{\nbrs_i} ) + g_i(x_i)$ 
  of problem~\eqref{eq:partitioned_problem} is a coercive function. 
  \oprocend
  \label{ass:coercive}
\end{assumption}

Assumption~\ref{ass:coercive} guarantees that at least a local 
minimum for problem~\eqref{eq:partitioned_problem} exists.

Figure~\ref{fig:partitions} visualizes the sparsity structure for a function
partitioned according to a path graph of $N = 4$ nodes.  Each $i$-th column
shows the variables on which $f_i$ depends, while along each $i$-th row it is
possible to see in which functions a variable $x_i$ appears. It is worth
noticing that the sparsity in the $i$-th row shows the consistency that needs to
be maintained among neighboring nodes on variable $x_i$.


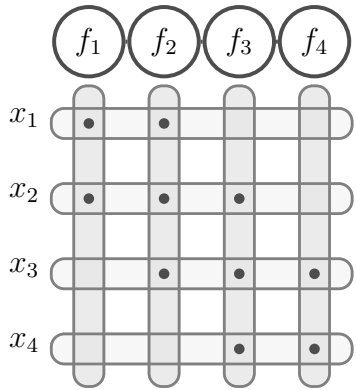
\begin{figure}[!ht]
  \centering
  \begin{tikzpicture}[scale=1]
%
%
%
    \def\width{0.2}
    \def\cornerWidth{5pt}
    
	  \foreach \y[evaluate=\y as \eval using int(5-\y)] in {1,...,4}
		  \draw[gray, line width=1pt, rounded corners=\cornerWidth, 
		            fill=gray!10,opacity=0.6]
		            (0.5,\eval - \width) rectangle (4.5, \eval + \width);

	  \foreach \x in {1,...,4}		
		  \draw[gray, line width=1pt, rounded corners=\cornerWidth,
		            fill=gray!25,opacity=0.6] 
		            (\x - \width , 0.5) rectangle (\x + \width, 4.5);

	  \foreach \y[evaluate=\y as \eval using int(5-\y)] in {1,...,4}
		  \draw[gray, line width=1pt, rounded corners=\cornerWidth]
		    (0.5,\eval - \width) node[above left,black]{$x_\y$} rectangle (4.5, \eval + \width);

	  \foreach \x in {1,...,4}		
      \draw[gray, line width=1pt, rounded corners=\cornerWidth]
          (\x - \width , 0.5) rectangle (\x + \width, 4.5) 
          node[above = 0.1cm, draw=black!70, line width=1.5pt,
            xshift=-0.2cm, minimum size = 0.5cm, circle, text=black] (f\x) {$f_{\x}$};

    \draw[line width=1.5pt, black!70] (f1) -- (f2);
    \draw[line width=1.5pt, black!70] (f2) -- (f3);
    \draw[line width=1.5pt, black!70] (f3) -- (f4);

    \fill[black!70] (4,1)  circle (2.pt);
    \fill[black!70] (1,3)  circle (2.pt);
    \fill[black!70] (1,4)  circle (2.pt);
	
    \fill[black!70] (2,2)  circle (2.pt);
    \fill[black!70] (2,3)  circle (2.pt);

    \fill[black!70] (3,3)  circle (2.pt);
    \fill[black!70] (3,1)  circle (2.pt);

    \fill[black!70] (3,2)  circle (2.pt);

    \fill[black!70] (4,2)  circle (2.pt);
    \fill[black!70] (2,4)  circle (2.pt);


\end{tikzpicture}

\caption{
  Partitioned optimization problem over a path graph of $N\!=\!4$ nodes.
  }
\label{fig:partitions}
\end{figure}


\section{Distributed optimization algorithm}
\label{sec:algorithm}
In this section we present our asynchronous distributed algorithm. 

In order to develop our algorithm, we need to introduce some 
technical tools: (i) the asynchronous communication protocol 
necessary to manage the overall behavior of the algorithm,
and (ii) the local approximation model that each node will use to perform
its local (descent) update.

We consider an asynchronous communication protocol where each node
$i\in\until{N}$ has its own concept of time defined by a local timer $\tau_i$,
which randomly and independently of the other nodes triggers when to awake
itself. The timers trigger according to exponential distributions with a common
parameter.
We denote $T_i$ a realization drawn by node $i$.
Between two triggering events the node is in an \emph{idle} mode, i.e., it
continuously receives messages from neighboring nodes and updates some internal
variables. When a trigger occurs, it switches into an \emph{awake} mode in which
it updates its local variable and transmits the updated information to its
neighbors.
A formal discussion on this protocol is given in~\cite{notarnicola2015randomized}.

The proposed distributed algorithm is based on \emph{local} quadratic, strongly-convex
approximations of the cost function that each node computes.

Formally, each node $i\in \until{N}$ constructs the following local
approximation of the entire cost function at a fixed $\bar{x} \in\real^n$
(neglecting the constant term $f(\bar{x})$ which does not affect the
optimization), 
\begin{align}
\label{eq:local_approximation_pb}
  q_i(s_i; \bar{x}) &\! := 
   \nabla_{x_i} f( \bar{x} ) ^\top s_i 
  + \frac{1}{2} \| s_i \|^2_{Q_i(\bar{x})} \! + g_i(\bar{x}_i + s_i)\\
  &\,=\sum_{j\in\nbrs_i} \!\! \nabla_{x_i} f_j( \bar{x}_{\nbrs_j} ) ^\top s_i 
  + \frac{1}{2} \| s_i \|^2_{Q_i(\bar{x})} \! + g_i(\bar{x}_i + s_i)  \notag
\end{align}
with $Q_i(x) \in \real^{n_i\times n_i}$ a symmetric, positive definite matrix
satisfying the following assumption. 
%
\begin{assumption}
  For any $x \in \real^n$ and $i\in\until{N}$ it holds that $Q_i(x)\succeq L_i I$.
  \oprocend
\label{ass:Hi}
\end{assumption}
Intuitively Assumption~\ref{ass:Hi} guarantees the strong convexity of
$q_i$. The role of the Lipschitz constant $L_i$ in the bound will be clear in the
analysis of the algorithm given in Section~\ref{sec:analysis}.

Informally, the asynchronous distributed optimization algorithms is as follows.
A node $i$ takes care of modifying the variable $x_i$. We denote $\sx{i}$ the
current state of node $i$, which is the estimated optimal value of the variable
$x_i$. Consistently we denote $\sx{\nbrs_i}$ the vector of states of nodes in
$\nbrs_i$.

When a node $i$ wakes up, it updates its state $\sx{i}$ by moving in the
direction obtained from the minimization of its local approximation
$q_i (s_i; \sx{})$, being $\sx{}$ the current value of the decision
variable. Then, it sends to each neighbor $j\in\nbrs_i$ the updated $x_i$ and
$\nabla_{x_j} f_i(\sx{\nbrs_i})$.
When in idle, node $i$ is in a listening mode. If an updated
$\nabla_{x_i} f_j(\sx{\nbrs_j})$ is received from a neighbor $j$ no computation
is needed. If $\sx{j}$ is also received ($j$ was an awake node) the following
happens. Node $i$ updates the partial gradients of its local function $f_i$
according to the new $\sx{j}$, and sends back the updated partial gradients to
its neighbors.
In order to highlight the difference between updated and old variables at node
$i$ during the awake phase, we denote the updated ones with a ``$+$'' symbol, e.g.,
as ${\sx{i}}^+$.

We want to stress two important aspects of the idle/awake cycle. First, these
two phases are regulated by local timers without the need of any central
clock. Second, when in idle a node only receives messages and from time to time
evaluates a partial gradient, which takes a negligible time compared to the
computation performed in the awake phase.

The distributed algorithm is formally reported in the table below (from the perspective of
node $i$). 
\begin{algorithm}
\renewcommand{\thealgorithm}{}
\floatname{algorithm}{Distributed Algorithm}
  \begin{algorithmic}[0] 
  
    \Statex \textbf{Processor state}:
    $\sx{i}$\medskip

    \Statex \textbf{Initialization}: set $\tau_i = 0$ and get a realization $T_i$ \medskip
    
    \Statex \textbf{Evolution}:
    \Label \texttt{\textbf{\textit{IDLE:}}}

      \StatexIndent[0.5] \textsc{while:} $\tau_i \leq T_i$ \textsc{do}:\medskip
      \StatexIndent[1] receive $\sx{j}$ and/or $\nabla_{x_i} f_j(\sx{\nbrs_j})$ from $j\in \nbrs_i$\medskip
      \StatexIndent[1] evaluate $\nabla_{x_j} f_i ( \sx{\nbrs_i} )$ and send it to $j\in\nbrs_i$
      \medskip

     \StatexIndent[0.5] go to \texttt{\textbf{\textit{AWAKE}}}.
      \vspace{0.2cm}

    \Label \texttt{\textbf{\textit{AWAKE:}}}

      \StatexIndent[0.5] \vspace{-0.8cm}

      \begin{flalign}
        \hspace*{0.8em}\text{compute} & \hspace{1cm} d_i = \argmin_{s_i} q_i (s_i; \sx{}) && 
        \label{eq:alg_descent}
      \end{flalign}
      \StatexIndent[0.5] \vspace{-0.6cm}
      \begin{flalign}
        \hspace*{0.8em}\text{update} & \hspace{1.7cm} \sx{i}^+ = \sx{i} + d_i && 
        \label{eq:alg_update}
      \end{flalign}      

     \StatexIndent[0.5] broadcast $\sx{i}^+$, $\nabla_{x_j} f_i (
     \sx{\nbrs_i}^+)$ to $j\in\nbrs_i$ \medskip
 
     \StatexIndent[0.5] set $\tau_i=0$, get a new realization $T_i$ and go to
     \texttt{\textbf{\textit{IDLE}}}.

  \end{algorithmic}
  \caption{Partitioned Coordinate Descent} 
  \label{alg:algorithm}
\end{algorithm}
%

We point out some aspects involving the local 
approximation~\eqref{eq:local_approximation_pb}
that each node uses in its local computations.

First, it is worth noting $q_i(s_i; \sx{})$ does not depend on the entire state
$\sx{}$, but only on $\sx{\nbrs_j}$, $j\in\nbrs_i$ and therefore is constructed
by node $i$ by using only information from its neighbors. Moreover, node $i$
does not needed the expression of neighboring cost functions $f_j$ to build
$q_i(s_i; \sx{})$, but only the gradients $\nabla_{x_i} f_j$. In some special
cases (discussed in the following paragraph), $Q_i(x)$ could include second order
information of $f_j$, $j\in\nbrs_i$, i.e., $\nabla_{x_i,x_i}^2 f_j$, that should
be sent together with the gradients. 

Second, different choices for the weight matrix $Q_i(x)$ are allowed.  
By exploiting the block Lipschitz continuity of the gradient of $f$, a 
first simple choice is to set $Q_i(x) := L_i I$ for all $i\in\until{N}$ and 
$x\in \real^n$. 
Motivated by existing works in the literature,
e.g.,~\cite{facchinei2015parallel}, non diagonal choices for $Q_i(x)$ are
reasonable: for instance, assuming $f\in \CC^2$, one can select a second order
approximation, i.e., set $Q_i(x):=\nabla^2_{x_i,x_i} f(x) + \epsilon_i I$ for a
sufficiently large $\epsilon_i>0$ for all $i\in\until{N}$.  As mentioned above
this information can be constructed in a distributed manner.

Third and final, recalling the definition of the proximal operator
$\map{\prox_{\alpha,\varphi}}{\real^n}{\real^n}$ of a closed, proper, convex
function $\map{\varphi}{\real^n}{\real\cup \{+\infty\}}$ given by
$\prox_{\alpha,\varphi} (v) := \argmin_{x} \big( \varphi(x) + \frac{1}{2\alpha}
\| x - v\|^2 \big)$
with $\alpha>0$, we have that for $Q_i(x)= L_i I$ the update law described
in~\eqref{eq:alg_descent}-\eqref{eq:alg_update}, can be rephrased in term of
proximal operators and, thus, leading to a distributed coordinate proximal
gradient method.
On this regard it is worth noting that our algorithm, with a general expression
for $Q_i$, can be written in terms of a generalized, weighted version of the
proximal operator as follows.
%
Given a positive definite matrix $W \in \real^{n \times n}$, we define
 \begin{align}
 \label{eq:gen_prox_operator}
   \prox_{W, \varphi} (v) := \argmin_x \Big\{ \varphi(x) +
     \frac{1}{2} \big\| x - v \big\|^2_{W^{-1}} \Big\},
 \end{align}
thus, 
the iteration~\eqref{eq:alg_descent}-\eqref{eq:alg_update} can be recast as
\begin{align*}
  \sx{i}^+ & = \prox_{Q_i(\sx{})^{-1}, g_i} \Big( \sx{i} - 
  Q_i(\sx{})^{-1} \sum_{j\in\nbrs_i} \!\! \nabla_{x_i} f_j( \sx{\nbrs_j} ) \Big).
\end{align*}

%

\section{Convergence analysis of the Partitioned Coordinate Descent distributed
  algorithm}
\label{sec:analysis}
In this section we prove the convergence in probability of the proposed
algorithm.

First, it is worth pointing out that being the algorithm asynchronous, for the
analysis we need to carefully formalize the concept of algorithm iterations.
We will use a nonnegative integer variable $t$ indexing a change in the whole
state $\sx{} = [\sx{1}^\top \ldots \sx{N}^\top]^\top$ of the distributed algorithm.
%
%
In particular, each triggering will induce an \emph{iteration} of the
distributed optimization algorithm and will be indexed with $t$. 
We want to stress that this (integer) variable $t$ 
does not need to be known by the agents.
That is, this timer is not a common clock and is only introduced for the
sake of analysis.

\begin{theorem}
  Let Assumptions~\ref{ass:Lipschitz_f_i}, \ref{ass:convex_g_i},
  \ref{ass:coercive} and~\ref{ass:Hi} hold true. Then, the Partitioned
  Coordinate Descent distributed algorithm generates a sequence
  $x(t):=[\sx{1}(t)^\top, \ldots, \sx{N}(t)^\top]^\top$ (obtained stacking the nodes'
  states) such that the random variable $V(x(t))$ converges almost surely, i.e.,
  there exists a random variable $V^\star$ such that
  \begin{align*}
    \Pr \Big( V(x(t) ) = V^\star \Big) = 1.
  \end{align*}

  Moreover, any limit point $x^\star$ of
  $[\sx{1}(t)^\top, \ldots, \sx{N}(t)^\top]^\top$ is a \emph{stationary} point
  of problem~\eqref{eq:partitioned_problem} and, thus, satisfies its first
  order optimality condition, i.e., there exists a subgradient $\widetilde\nabla g ( x^\star )$
  of $g$ at $x^\star$ such that $\nabla f( x^\star ) + \widetilde\nabla g ( x^\star ) = 0$.  \oprocend
  %
  \label{thm:nonconvex_coordinate}%
\end{theorem}%
%

\subsection{Coordinate descent method for composite non-convex minimization}
\label{sec:coordinate_method}
In this subsection we consider a more general composite optimization problem and
prove a result that is instrumental to the convergence proof of our distributed
algorithm. We introduce a generalization of the algorithm proposed
in~\cite{nesterov2013gradient,richtarik2014iteration,patrascu2015efficient}
based on the quadratic approximation introduced in~\eqref{eq:local_approximation_pb}.
We present the algorithm for problem~\eqref{eq:partitioned_problem}, but we want
to stress that the algorithm can be applied to a general function
$\map{f}{\real^n}{\real}$ with block-Lipschitz continuous gradient.
This will be clear from the analysis.

%

We consider a coordinate descent method based on selecting a random 
block-component, say $x_i$, of $x$ at each iteration and updating 
only $x_i$ through a suitable descent rule. 
The descent step is based on the quadratic approximation of the cost function
given in \eqref{eq:local_approximation_pb}.
The coordinate descent method is formally summarized in the table below. 
\begin{algorithm}[H]
  \renewcommand{\thealgorithm}{}
  \begin{algorithmic}[0] 
    \StatexIndent[-0.3] Choose a random block $i_t \! \in \! \until{N}$ with probability
    $p_{i_t}$

    \StatexIndent[-0.3] Compute a descent direction $d_{i_t}$ solving
    \begin{align}
    \label{eq:cd_descent_dir}
      d_{i_t} = \argmin_{s_i} \, q_{i_t} (s_i;x(t))
    \end{align}
    
    \StatexIndent[-0.3] Update the decision variable according to
      \begin{align}
      \label{eq:cd_update}
      \begin{split}
        x_{i_t}(t+1) & = x_{i_t}(t) + d_{i_t} 
        \\
        x_j(t+1) & = x_j(t), \hspace{1cm} \text{for all } j \neq {i_t}
      \end{split}
      \end{align}

  \end{algorithmic}
  \caption{Generalized Coordinate Descent Algorithm} 
  \label{alg:nonconvex_CD}
\end{algorithm}

In the following we present results for the theoretical convergence
of the generalized coordinate descent algorithm.

\begin{lemma}
Let Assumption~\ref{ass:Lipschitz_f_i}, \ref{ass:convex_g_i}, \ref{ass:Hi} hold.
Let $x(t)$ be the random sequence generated by Generalized Coordinate Descent Algorithm, 
then for all $t\ge 0$ it holds
\begin{align*}
  V(x(t+1)) \le V(x(t)) - \frac{L_{i_t} }{2} \|d_{i_t} \|^2.
\end{align*}
\label{lem:descent_cost}
\end{lemma}
%

%
\begin{proof}
  From Assumption~\ref{ass:Lipschitz_f_i} (Lipschitz continuity of $\nabla f$),
  we can write the well-known descent lemma (see \cite[Proposition
  A.24]{bertsekas1999nonlinear}), 
  for all $i\in\until{N}$ and for all $\bar{x}\in\real^n$
\begin{align*}
\begin{split}
  V( \bar{x}+U_is_i) & \le f( \bar{x} ) + \nabla_{x_i} f( \bar{x} ) ^\top s_i 
  \\
  & \hspace{0.4cm} + \frac{L_i}{2} \|s_i\|^2 + g_i(\bar{x}_i+s_i) + \textstyle \sum_{j\neq i} g_j(\bar{x}_j),
\end{split}
\end{align*}
with $U_i$ introduced in the \emph{Notation} paragraph.


Since $Q_i$ satisfies Assumption~\ref{ass:Hi}, then we can generalize the above
descent condition by introducing a uniform bound depending
on the Lipschitz constant of block $i$, i.e.,
\begin{align*}
\begin{split}
  V( \bar{x} +U_is_i) & \le f( \bar{x} ) + \nabla_{x_i} f( \bar{x} ) ^\top s_i 
  \\
  & + \frac{1}{2} \|s_i\|^2_{Q_i(\bar{x})} \!+\! g_i(\bar{x}_i+s_i) \!+\! \textstyle \sum_{j\neq i} g_j(\bar{x}_j)
\end{split}
\end{align*}

Due the partitioned structure of $f$, the explicit expression of $\nabla_{x_i} f(x)$ 
actually depends only on $f_j$, $j \in \nbrs_i$, thus the latter condition can be further 
rephrased as
\begin{align}
  V(\bar{x}+U_is_i) \leq q_i(s_i ; \bar{x} ) + f(\bar{x}) + \textstyle \sum_{j\neq i} g_j(\bar{x}_j).
\label{eq:descent_pb}
\end{align}
with $q_i(s_i ; \bar{x} )$ defined as in~\eqref{eq:local_approximation_pb}.

Consider a descent direction $d_{i_t}$ computed as in~\eqref{eq:cd_descent_dir},
then $d_{i_t}$ satisfies the first order necessary condition of optimality 
for problem~\eqref{eq:cd_descent_dir}
\begin{align}
\label{eq:FNC}
  \nabla_{x_{i_t}} f(x(t)) \!+\! Q_{i_t}(x(t)) d_{i_t} 
  \!+\! \widetilde{\nabla} g_{i_t} ( x_{i_t} (t)+d_{i_t}) = 0,
\end{align}
where $\widetilde{\nabla} g_{i_t} \in\real^{m_{i_t}}$ is a particular subgradient of $g_{i_t}$.

Starting form equation~\eqref{eq:descent_pb} with the following identification 
$\bar{x} = x(t)$ and $\bar{x}+U_{i_t} d_{i_t} = x(t+1)$, and adding and 
subtracting the term $g_{i_t} (x_{i_t}(t) )$ we obtain
\begin{align*}
\begin{split}
  V(x(t+1) ) & \le 
  V ( x(t) ) + \nabla_{x_{i_t}} f ( x(t) ) ^\top\! d_{i_t} \!\!+\! \frac{1}{2} \|d_{i_t}\|^2_{Q_{i_t} (x(t))}
  \\
  &\hspace{0.5cm}  + g_{i_t} (x_{i_t} (t) + d_{i_t} ) - g_{i_t} (x_{i_t} (t)) 
  \\
  & \le V ( x(t) ) + \nabla_{x_{i_t}} f ( x(t) ) ^\top d_i 
  \\
  &\hspace{0.5cm} + \frac{1}{2} \|d_{i_t}\|^2_{Q_{i_t}(x(t))} + \widetilde{\nabla} g_{i_t} (x_{i_t}(t) +d_{i_t})^\top d_{i_t}
  \\
  & \le V ( x(t) ) - \frac{1}{2} \|d_{i_t}\|^2_{Q_{i_t} (x(t))}
  \\
   & \le V ( x(t) ) - \frac{L_i}{2} \|d_{i_t}\|^2
\end{split}
\end{align*}
where we used the convexity of $g_{i_t}$, the optimality condition~\eqref{eq:FNC} and the uniform bound
in Assumption~\ref{ass:Hi}.
\end{proof}

\begin{theorem}
  Let Assumptions~\ref{ass:Lipschitz_f_i}, \ref{ass:convex_g_i},
  \ref{ass:coercive} and~\ref{ass:Hi} hold true.  Then, the Generalized
  Coordinate Descent Algorithm generates a sequence $x(t)$ such that the random
  variable $V(x(t))$ converges almost surely.  Moreover, any limit point
  $x^\star$ of $x(t)$ is a stationary point of $V$ and, thus, satisfies the
  first order necessary condition for optimality for
  problem~\eqref{eq:partitioned_problem}, i.e., there exists a subgradient
  $\widetilde \nabla g ( x^\star )$ of $g$ at $x^\star$ such that
  \begin{align*}
    \nabla f( x^\star ) + \widetilde\nabla g ( x^\star ) = 0 
  \end{align*}%
  \label{thm:theorem_nonconvex_coordinate}%
\end{theorem}
\begin{proof}
  The result is proven by following the same line as in~\cite[Theorem~1]{patrascu2015efficient} 
  where the generalized Lemma~\ref{lem:descent_cost} is used in place 
  of~\cite[Lemma~3]{patrascu2015efficient}.
\end{proof}

\subsection{Proof of Theorem~\ref{thm:nonconvex_coordinate}}
\label{sec:theorem_proof}
Our proof strategy is based on showing that the iterations of the asynchronous
distributed algorithm can be written as the iterations of an ad-hoc version of
the coordinate descent method for composite non-convex functions given in
Section~\ref{sec:coordinate_method}.


\emph{Timer model and uniform node extraction.}
%
Since the timers trigger independently according to the same exponential
distribution, then from an external, global perspective, the induced awaking
process of the nodes corresponds to the following: only one node per iteration
wakes up randomly, uniformly and independently from previous iterations.
Thus, each triggering, which induces an \emph{iteration} of the distributed
optimization algorithm and is indexed with $t$, corresponds to the (uniform)
selection of a node in $\until{N}$ that becomes awake.
We denote $i_t$ the extracted node. Notice that node $i_t$ changes the value of
its state $\sx{i_t}$ while all the other states are not changed by the
algorithm.


\emph{State consistency (inductive argument).} Next we show by induction that if
all the nodes have a consistent and updated information before a node $i$ gets
awake, then the same holds after the update. By consistent we mean that for a
variable $x_\ell$, all the nodes in $\nbrs_\ell$ have the same state
$\sx{\ell}$. By updated we mean that each node $\ell$ has an updated value of
the gradients $\nabla_{x_\ell} f_j$, $j\in\nbrs_\ell$.
First, node $i$ changes only its state $\sx{i}$ relative to the variable
$x_i$. This variable is shared only with neighbors $j\in\nbrs_i$, which receive the new
state $\sx{i}$ after the update. 
As regards the gradients, the ones affected by the change of the variable $x_i$
are $\nabla_{x_i} f_j$, with $j\in\nbrs_i$. Notice that these
gradients are only used by nodes $k\in\nbrs_j$. But after the broadcast
performed by $i$, each idle $j\in\nbrs_i$ receives the updated $\sx{i}$,
updates the gradients, and sends them to its neighbors $k\in\nbrs_j$. 
The variables and gradients for the rest of the nodes in the network are not
changed by the update of node $i$.


\emph{Coordinate descent equivalence and convergence.} Finally, we simply notice
that, thanks to the consistency argument just shown, steps
\eqref{eq:alg_descent}-\eqref{eq:alg_update} correspond to steps
\eqref{eq:cd_descent_dir}-\eqref{eq:cd_update}. 
Thus, we have shown that our distributed algorithm implements the centralized
coordinate method and therefore inherits its convergence properties.  By
invoking Theorem~\ref{thm:theorem_nonconvex_coordinate}, the proof follows.






\section{Numerical simulations on a non-convex constrained quadratic program}
\label{sec:simulations}

In this section we present a numerical example showing the 
effectiveness of the proposed algorithm.

We consider an undirected connected Erd\H{o}s-R\'enyi random 
graph $\GG$, with parameter $0.2$,  connecting $N=50$ nodes 
and we test the distributed algorithm on a partitioned non-convex 
constrained quadratic program in the form
\begin{align}
  \min_{x\in\real^n} &\: \sum_{i=1}^N x_{\nbrs_i}^\top 
    H_i x_{\nbrs_i} + r_i^\top x_{\nbrs_i}  + I_{X_i}(x_i),
  \label{eq:numerical_problem}
\end{align}
where each $x_i\in \real$ for all $i\in\until{N}$ and each cost matrix
$H_i \in \real^{|\nbrs_i|\times |\nbrs_i|}$ is only symmetric (\emph{not}
positive definite). We construct $H_i$ as the difference between a positive
definite matrix $\tilde{H}_i \in \real^{|\nbrs_i|\times |\nbrs_i|}$ and a
suitable scaled version of the identity matrix.  Finally, each function
$I_{X_i}$ denotes the indicator function of the segment $X_i = [-\ell_i,u_i]$,
i.e., we constrain each $x_i$ to lie into an interval.  We set $\ell_i = -30$
and $u_i = 20$ for all $i\in\until{N}$.

Problem~\eqref{eq:numerical_problem} fits our set-up described in 
Section~\ref{sec:setup} by defining
\begin{align*}
  f_i ( x_{\nbrs_i} ) := x_{\nbrs_i} ^\top H_i x_{\nbrs_i} + r_i^ \top x_{\nbrs_i} 
\end{align*}
and
\begin{align*}
  & g_i (x_i) := I_{X_i}(x_i) = 
    \begin{cases}
    x_i & \text{ if } \ell_i \le x_i \le u_i
    \\
    +\infty & \text{ otherwise.}
  \end{cases}
\end{align*}
Moreover, we use the local approximation $q_i(s_i,\sx{} )$
as in~\eqref{eq:local_approximation_pb} with 
$Q_i = \frac{1}{\alpha_i} I $
with $\alpha_i=0.01$ for all $i\in\until{N}$.

In Figure~\ref{fig:x} we plot the evolution of two selected components 
of the decision variable $x$ at each iteration $t$ (defined as discussed 
in Section~\ref{sec:analysis}), i.e., $x_i (t)$, $i = 14, 48$. The horizontal dotted 
lines represent the centralized solution.
Since the algorithm is asynchronous and based on a coordinate approach, we plot 
the rate of convergence with respect to the normalized iterations $t/N$ in order to show
the effective behavior with respect to the global time.
\begin{figure}[!htbp]
\centering
  \includegraphics[scale=0.8]{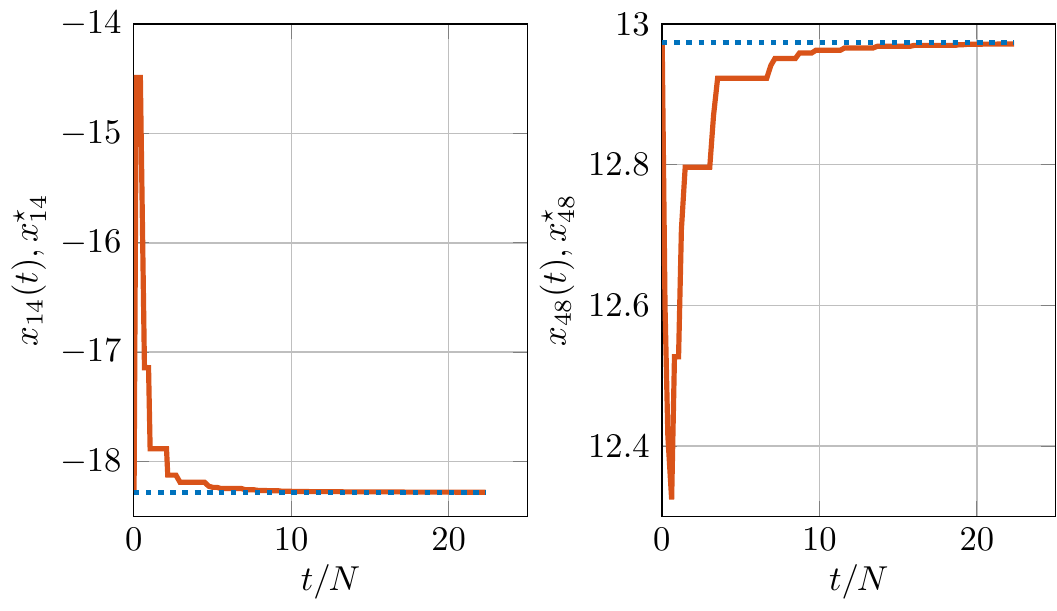}
  \caption{
    Evolution of two decision variables $x_i$, $i=14,48$, for 
    the distributed algorithm.
    }
  \label{fig:x}
\end{figure}

In Figure~\ref{fig:cost} we show the difference, in logarithmic scale, 
between the cost $V(x(t))$ at each iteration $t$ and the value of $V$ 
attained at the limit point $x^\star$ of $x(t)$ (proven to be a stationary point). 
\begin{figure}[!htbp]
\centering
  \includegraphics[scale=0.8]{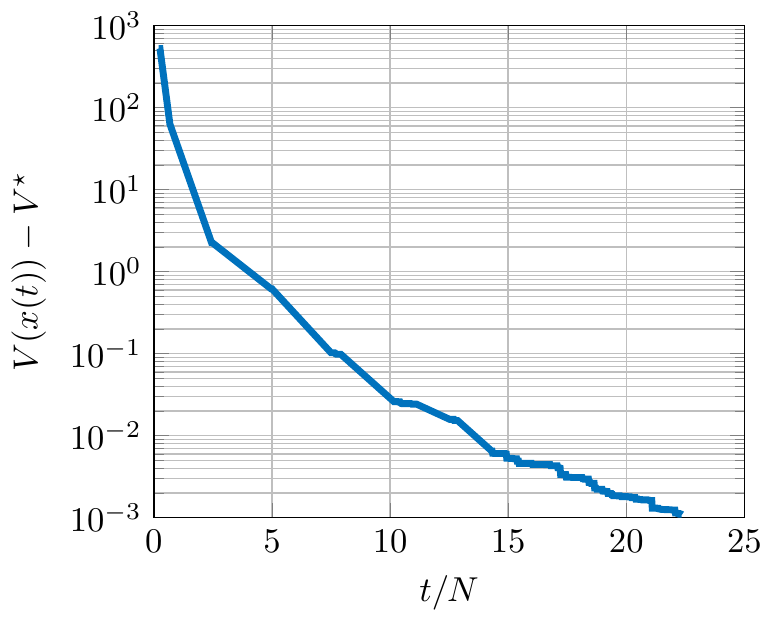}
  \caption{
    Evolution of the cost error, in logarithmic scale, for the distributed algorithm.
    }
  \label{fig:cost}
\end{figure}


\section{Conclusions}
\label{sec:conclusions}

In this paper we have proposed an asynchronous, distributed algorithm to solve
partitioned, big-data non-convex optimization problems.  The main idea is
that each node updates its local variable by minimizing a suitable, local
quadratic approximation of the cost, built via an information exchange with
neighboring nodes.  We prove the convergence of the distributed algorithm by
showing that it corresponds to a proper instance of a coordinate descent method.


\section*{Acknowledgments}
The authors would like to thank Angelo Coluccia e Massimo Frittelli for their
help and suggestions.



  \bibliographystyle{IEEEtran}
  \bibliography{proximal} 

\end{document}